\documentclass[12pt,english]{article}
\usepackage{a4}
\usepackage[round]{natbib}
\usepackage{babel}

\usepackage[margin=1.1in]{geometry}
\usepackage{array}
\usepackage{amsmath,amssymb}
\usepackage{amsthm}
\usepackage{amsbsy}
\usepackage[colorlinks = true,linkcolor = blue,urlcolor  = blue,citecolor = blue,anchorcolor = blue]{hyperref}
\usepackage{setspace}
\usepackage{xfrac}
\usepackage{tikz}
\usepackage[affil-it]{authblk}

\begin{document}

\title{Learning to agree over large state spaces\footnote{I thank Hannu Vartiainen for invaluable support and supervision. I also thank Hannu Salonen, Mark Voorneveld, Klaus Kultti, Aviad Heifetz, Mats Godenhielm and two anonymous reviewers for helpful comments that greatly improved the paper, and the audiences of the 41st Annual Meeting of the Finnish Economic Association, the 15th European Meeting on Game Theory, Bayes By the Sea Conference 2019 and the Lisbon Meetings in Game Theory and Applications 2019. Financial support from the OP Group Research Foundation, the Yrj\"{o} Jahnsson Foundation, the Finnish Cultural Foundation and the University of Helsinki is gratefully acknowledged. All errors are mine.}}

\author{Michele Crescenzi\\michele.crescenzi@helsinki.fi}
\affil{\small University of Helsinki and Helsinki Graduate School of Economics, Finland}

\date{January 2022\\
\vspace*{10pt}
\textcolor{red}{This version is published in the \textit{Journal of Mathematical Economics} (2022), DOI: \url{https://doi.org/10.1016/j.jmateco.2022.102654}}}

\newtheorem{proposition}{Proposition}
\newtheorem{theorem}{Theorem}
\newtheorem*{theorem*}{Theorem}
\newtheorem*{definition*}{Definition}
\newtheorem{assumption}{Assumption}
\newtheorem{definition}{Definition}
\newtheorem{lemma}{Lemma}
\newtheorem{claim}{Claim}
\newtheorem{remark}{Remark}
\newtheorem*{remark*}{Remark}
\newtheorem{corollary}{Corollary}

\maketitle
\onehalfspacing

\begin{abstract}
We study how a consensus emerges in a finite population of like-minded individuals who are asymmetrically informed about the realization of the true state of the world. Agents observe a private signal about the state and then start exchanging messages. Generalizing the classical model of rational dialogues of \cite{geanakoplos1982} and its subsequent extensions, we dispense with the standard assumption that the state space is a probability space and we do not put any bound on the cardinality of the state space itself or the information partitions. We show that a class of rational dialogues can be found that always lead to consensus provided that three main conditions are met. First, everybody must be able to send messages to everybody else, either directly or indirectly. Second, communication must be reciprocal. Finally, agents need to have the opportunity to engage in dialogues of transfinite length.

\bigskip

\bigskip

JEL CLASSIFICATION: C70, D82, D83

\end{abstract}

\newpage

\section{Introduction}

This is a paper on common knowledge acquisition and consensus. We are interested in the following problem. Consider a situation of incomplete information with finitely many agents. Each agent privately observes a partitional signal, whose realization depends on the true state of the world. Afterwards, agents can exchange messages, whose content depends on the private information each agent has. What are the conditions under which everyone eventually sends the same message? In other words, what are the conditions under which a consensus eventually emerges? And when a consensus is attained, is it common knowledge that everyone agrees?

Originally investigated by \cite{geanakoplos1982}, the problem of rational dialogues leading to consensus has been studied in two distinct, yet not mutually exclusive, settings. First, in probabilistic environments like \cite{geanakoplos1982}, the state space is a probability space and the messages exchanged are posterior probabilities about a fixed event. Second, in more general environments like \cite{bacharach1985} or \cite{krasucki1996}, the state space is not necessarily a probability space, and the messages exchanged are values of some abstract function. In addition, agents' information partitions are assumed to be finite or, as in \cite{parikh1992infinite}, countably infinite.

Our goal is to study rational dialogues and consensus in a strictly more general setting in which the state space is just an arbitrary set, the messages exchanged are values of some abstract function, there are no bounds on the cardinality of the state space itself or the information partitions, and the length of a dialogue can be any ordinal number. Differently put, we study the emergence of consensus starting from a minimalist setting where we have a non-empty set as the state space, an information partition for each agent, and a well-defined function whose values are communicated in an arbitrarily long dialogue. We assume that agents communicate truthfully, do not tell lies and are not driven by strategic motives. The reason for exploring such a generalization is that, strictly speaking, the concepts of knowledge, common knowledge and consensus do not call for a probabilistic structure of the state space. Nor do they require a particular bound on the cardinality of the state space itself or the information partitions. Of course, limitations of this kind can be needed in particular cases, for example when the dialogue is about posterior probabilities. Nonetheless, such limitations are not necessary to formulate the general question of when a dialogue between rational agents leads to consensus.

Our contributions are the following. First, we provide a general framework for the study of rational dialogues that induce a consensus. This framework consists of finitely many agents, an arbitrarily large state space, an abstract function whose value are communicated in dialogues of transfinite length, and private communication. Previous results in the literature can be obtained as particular cases of our model. Second, we give sufficient conditions for the emergence of consensus. More specifically, provided that the function whose values are communicated is sufficiently regular, a consensus emerges if the communication protocol has two properties. Everyone should be able to talk to everyone else during a dialogue, either directly or indirectly. And communication should be reciprocal: if agent $i$ sends her message to $j$ at some point during a dialogue, then $j$ has to send a message back to $i$. These two conditions on the protocol are a natural generalization of those already provided by \cite{krasucki1996} for finite models. Finally, we write the model, and prove the main results, in such a way to exploit the lattice structure of the set of all partitions of the state space. The standard approach in the literature is to take a communication protocol as a primitive. Any such protocol induces a graph over the set of agents, and conditions ensuring consensus are found by studying the properties of this graph. Our approach takes the opposite route. The primitive object is a graph which describes who talks to whom during one round of communication. This graph induces a self-function in the set of profiles of information partitions: intuitively, it maps the information that agents have at the beginning of the communication round to the refined information they have after having talked. By iterating this function, we can generate a sequence of profiles of partitions that captures all the information generated during a dialogue. In so doing, we can work directly in the complete lattice of information partitions and prove convergence to consensus with relative ease.

As we mentioned, we allow dialogues to have transfinite length. The reason is that, when the state space is infinite, agents might have to exchange infinitely many messages before reaching a commonly known consensus. This point is illustrated in the example in Section \ref{sec:Example}. Now, the question arises: How to make sense of a dialogue whose length is indexed by some infinite ordinal? From a mathematical point of view, the concept of an infinitely long dialogue is fully legitimate and well-defined. Moreover, such dialogues have already been studied in \cite{parikh1992infinite} and \cite{aumann-hart2003}. The first paper explores common knowledge acquisition and consensus in a framework akin to ours, but with only two agents and a countably infinite state space. The other paper studies games played at the end of an infinitely long dialogue whereas we do not explore game-theoretic situations. Nonetheless, we are aware that the concept of a transfinite dialogue may sound fanciful and we hope the following points will partly overcome skepticism. First, we emphasize that we do not aim for practicality in this paper. Rather, our goal is to explore the formal structure of dialogues that, at least in principle, are conducive to a commonly known consensus. Second, dialogues of any length can be seen as interactive inferential processes, in which everyone draws conclusions from what everyone says or does, everyone draws conclusions from everyone's conclusions drawn from what everyone says or does, and so on \textit{ad infinitum}. Under this interpretation, we can phrase the key question of our paper as follows: How long should an interactive inferential process be in order for rational people to attain a commonly known consensus? An analogous question has been studied in game-theoretical contexts. Specifically, \cite{lipman1994} shows that in some games it is possible to attain common knowledge of rationality only through a transfinite process of elimination of never-best replies. Finally, we do not claim that transfinite dialogues are necessary in every possible situation covered by our model. As we show in the example in Section \ref{sec:Example}, transfinite dialogues become necessary when the message function through which agents communicate has a low expressive power and, consequently, not much information can be learned at each stage of the dialogue.

The rest of the paper is organized as follows. The next section contains an example of a rational dialogue in which it takes infinitely many stages to reach a commonly known consensus. The model is presented in Section \ref{sec:Model1} and the results are illustrated in Section \ref{sec:Results1}. Finally, in Section \ref{sec:Disc} we discuss the related literature and some of the assumptions we make in our analysis.

\section{Example}\label{sec:Example}
Consider a situation of incomplete information. There are three persons: Ann ($A$), Bob ($B$) and Carol ($C$). The set of possible states of the world is the set of natural numbers $\mathbb{N}$. The true state cannot be observed directly. Each person has access to a partitional signal whose realizations are privately observed. The resulting information partitions are as follows.
\begin{align*}
g^0_A &= \left\lbrace \{1,2,7\}, \{3,4\}, \{5\} \right\rbrace \cup \left\lbrace \{6+4k_1, 9+4k_1\}:k_1\geq 0\right\rbrace \cup \left\lbrace \{8+4k_2, 11+4k_2\}:k_2\geq 0\right\rbrace\\
g^0_B &= \left\lbrace \{1+4k_1, 2+4k_1\}:k_1\geq 0\right\rbrace \cup \left\lbrace \{3+4k_2, 4+4k_2\}:k_2\geq 0\right\rbrace\\
g^0_C &= \left\lbrace \{1+4k_1:k_1\geq 0\}, \{2+4k_2:k_2\geq 0\}, \{3+4k_3:k_3\geq 0\}, \{4+4k_4:k_4\geq 0\}\right\rbrace.
\end{align*}

A graphical representation is provided in Figure \ref{fig:Example}.

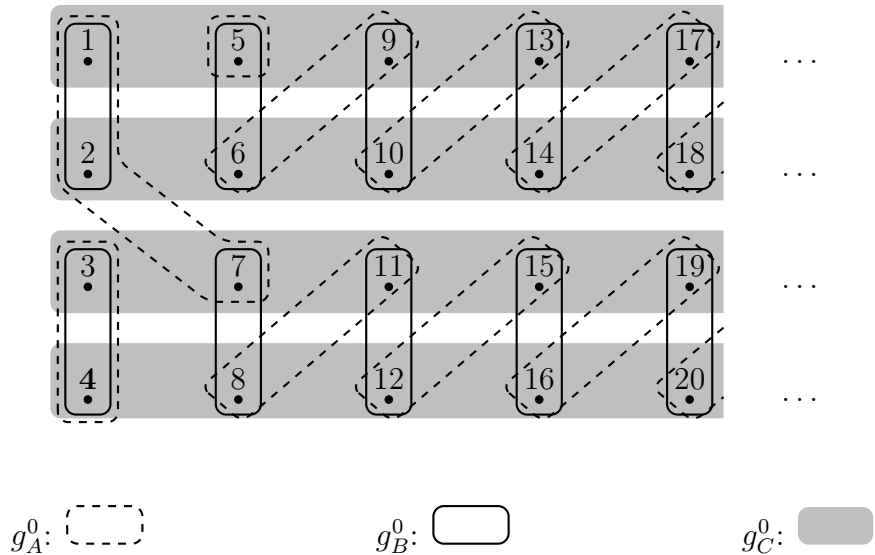
\begin{figure}[h]
\centering
	\begin{tikzpicture}
\path [rounded corners, dotted, thick, fill=lightgray] (8.45,5.25)--(-0.5, 5.25)--(-0.5,4.15)--(8.45,4.15);
\path [rounded corners, dotted, thick, fill=lightgray] (8.45,3.75)--(-0.5, 3.75)--(-0.5,2.65)--(8.45,2.65);
\path [rounded corners, dotted, thick, fill=lightgray] (8.45,2.25)--(-0.5, 2.25)--(-0.5,1.15)--(8.45,1.15);
\path [rounded corners, dotted, thick, fill=lightgray] (8.45,0.75)--(-0.5, 0.75)--(-0.5,-0.25)--(8.45,-0.25);
	\draw [fill] (0,0) circle [radius=0.05];
	\draw [fill] (2,0) circle [radius=0.05];
	\draw [fill] (4,0) circle [radius=0.05];
	\draw [fill] (6,0) circle [radius=0.05];	
	\draw [fill] (8,0) circle [radius=0.05];
	\draw [fill] (0,1.5) circle [radius=0.05];
	\draw [fill] (2,1.5) circle [radius=0.05];
	\draw [fill] (4,1.5) circle [radius=0.05];
	\draw [fill] (6,1.5) circle [radius=0.05];	
	\draw [fill] (8,1.5) circle [radius=0.05];
	\draw [fill] (0,3) circle [radius=0.05];
	\draw [fill] (2,3) circle [radius=0.05];
	\draw [fill] (4,3) circle [radius=0.05];
	\draw [fill] (6,3) circle [radius=0.05];	
	\draw [fill] (8,3) circle [radius=0.05];
	\draw [fill] (0,4.5) circle [radius=0.05];
	\draw [fill] (2,4.5) circle [radius=0.05];
	\draw [fill] (4,4.5) circle [radius=0.05];
	\draw [fill] (6,4.5) circle [radius=0.05];	
	\draw [fill] (8,4.5) circle [radius=0.05];
	\node [above] at (0,0) {$\mathbf{4}$};
	\node [above] at (0,1.5) {$3$}; 
	\node [above] at (0,3) {$2$};
	\node [above] at (0,4.5) {$1$};
	\node [above] at (2,0) {$8$};
	\node [above] at (2,1.5) {$7$};
	\node [above] at (2,3) {$6$};
	\node [above] at (2,4.5) {$5$};
	\node [above] at (4,0) {$12$};
	\node [above] at (4,1.5) {$11$};
	\node [above] at (4,3) {$10$};
	\node [above] at (4,4.5) {$9$};
	\node [above] at (6,0) {$16$};
	\node [above] at (6,1.5) {$15$};
	\node [above] at (6,3) {$14$};
	\node [above] at (6,4.5) {$13$};
	\node [above] at (8,0) {$20$};
	\node [above] at (8,1.5) {$19$};
	\node [above] at (8,3) {$18$};
	\node [above] at (8,4.5) {$17$};
	\node at (9.5,0) {$\dots$};
	\node at (9.5,1.5) {$\dots$};
	\node at (9.5,3) {$\dots$};
	\node at (9.5,4.5) {$\dots$};
	\draw[ rounded corners, thick] (-0.3,5) rectangle (0.3,2.8);
	\draw [rounded corners, thick] (1.7,5) rectangle (2.3,2.8);
	\draw [rounded corners, thick] (3.7,5) rectangle (4.3,2.8);
	\draw[ rounded corners, thick] (5.7,5) rectangle (6.3,2.8);
	\draw [rounded corners, thick] (7.7,5) rectangle (8.3,2.8);
	\draw [rounded corners, thick] (-0.3,2) rectangle (0.3,-0.2);
	\draw [rounded corners, thick] (1.7,2) rectangle (2.3,-0.2);
	\draw [rounded corners, thick] (3.7,2) rectangle (4.3,-0.2);
	\draw [rounded corners, thick] (5.7,2) rectangle (6.3,-0.2);
	\draw [rounded corners, thick] (7.7,2) rectangle (8.3,-0.2);
	\draw [rounded corners, dashed, thick] (-0.4,4.9)--(-0.4,2.85)--(1.6,1.3)--(2.4,1.3)--(2.4,2.1)--(1.8,2.1)--(0.4,3.185)--(0.4,5.1)--(-0.4,5.1)--(-0.4,4.5); 
	\draw [rounded corners, dashed, thick] (1.6,4.3) rectangle (2.4,5.1);
	\draw [rounded corners, dashed, thick] (-0.4,-0.3) rectangle (0.4,2.1);
	\draw [rounded corners, dashed, thick] (1.6,3.0875)--(2.05,2.7)--(4.45,4.75)--(3.9,5.215)--(1.5,3.165)--(1.6,3.0875);
	\draw [rounded corners, dashed, thick] (3.6,3.0875)--(4.05,2.7)--(6.45,4.75)--(5.9,5.215)--(3.5,3.165)--(3.6,3.0875);
	\draw [rounded corners, dashed, thick] (5.6,3.0875)--(6.05,2.7)--(8.45,4.75)--(7.9,5.215)--(5.5,3.165)--(5.6,3.0875);
	\draw [rounded corners, dashed, thick] (8.45,95/120 + 3.165)--(7.5,3.165)--(8.05,2.7)--(8.45,1/3+2.7);
	\draw [rounded corners, dashed, thick] (1.6,0.0875)--(2.05,-0.3)--(4.45,1.75)--(3.9,2.215)--(1.5,0.165)--(1.6,0.0875);
	\draw [rounded corners, dashed, thick] (3.6,0.0875)--(4.05,-0.3)--(6.45,1.75)--(5.9,2.215)--(3.5,0.165)--(3.6,0.0875);
	\draw [rounded corners, dashed, thick] (5.6,0.0875)--(6.05,-0.3)--(8.45,1.75)--(7.9,2.215)--(5.5,0.165)--(5.6,0.0875);
	\draw [rounded corners, dashed, thick] (8.45,95/120 + 0.165)--(7.5,0.165)--(8.05,-0.3)--(8.45,1/3-0.3);
	\end{tikzpicture}

\vspace{30pt}

$g^0_A$: 
\begin{tikzpicture}
\draw [rounded corners, dashed, thick] (0,0) rectangle (1,0.5);
\end{tikzpicture}
\hspace{80pt}
$g^0_B$: 
\begin{tikzpicture}
\draw [rounded corners, thick] (0,0) rectangle (1,0.5);
\end{tikzpicture}
\hspace{80pt}
$g^0_C$: 
\begin{tikzpicture}
\path [rounded corners, dotted, thick, fill=lightgray] (0,0) rectangle (1,0.5);
\end{tikzpicture}
\caption{Initial information partitions.}
\label{fig:Example}
\end{figure}

Suppose the true state of the world is $4$. This means that both $A$ and $B$ know that the true state can only be either $3$ or $4$, whereas $C$ knows that the true state can only belong to $\left\lbrace 4, 8, 12, \dots \right\rbrace$. If everyone publicly discloses his or her information block, each person will readily learn that the true state is $4$. But can people find out the true state when they disclose less than what they know? To be more concrete, suppose that each person sends messages according to the function $f$, which is defined as follows. For any information block $S\subseteq \mathbb{N}$,
\begin{equation*}
f(S) =
\begin{cases}
k & \text{if } S=\{k\} \text{ for some } k\in \mathbb{N}\\
0 & \text{otherwise}.
\end{cases}
\end{equation*}

In words, when someone knows that the true state is $k$, he or she sends the message $k$, which is a shorthand for the statement ``I know that the true state is $k$''. In all other cases, the message $0$ is sent. The latter is a shorthand for ``I do not know the true state''. This formulation is a variant of the classical puzzle of the muddy children (or unfaithful husbands or red and white hats). See, e.g., \cite{geanakoplos1994Handbook}.

We also assume that communication unfolds in stages. There is no limit on the number of stages and dialogues of transfinite length are feasible. In each stage, messages are privately exchanged. In particular, $B$ sends a message to both $A$ and $C$, and each of $A$ and $C$ sends a message to $B$.

Let's see what happens during the first stage of communication. Since nobody knows the true state, everyone sends the message $0$. Clearly, nobody learns the true state after messages are exchanged, and each person's information block containing $4$ stays the same. However, $B$ can refine his information partition at the end of the first stage of communication. The reason is this. If the true state were $5$ or $6$, $B$'s block containing the true state would be $\{5,6\}$. In addition, $A$ would send the message $5$ to $B$ if the true state were $5$, and she would send a different message, namely $0$, if the state were $6$. In either case, $B$ would learn the true state. Formally, this means that, at the end of the first communication stage, $B$ can split the block $\{5,6\}$ in his information partition into two singletons, $\{5\}$ and $\{6\}$. Differently put, $B$ is able to discriminate between $5$ and $6$ after one round of communication. We emphasize that this updating can be done even when, as we are assuming, the true state is $4$. In fact, a partition represents the information a person has not just at the true state but at every possible state of the world. As a result, people have the following information partitions at the end of the first stage of communication:
\begin{align*}
g^1_A &= g^0_A\\
g^1_B &= \left( g^0_B\backslash \{\{5,6\}\} \right) \cup \left\lbrace \{5\}, \{6\}\right \rbrace\\
g^1_C &= g^0_C.
\end{align*}

In the second stage, everybody sends the message $0$. In addition, $A$ splits her block $\{6,9\}$ into the two singletons $\{6\}$ and $\{9\}$. $C$ splits $\{1+4k_1:k_1\geq 0\}$ into $\{5\}$ and $\{1+4k_1:k_1 = 0 \textsf{ or } k_1 \geq 2\}$, and splits $\{2+4k_2:k_2\geq 0\}$ into $\{6\}$ and $\{2+4k_2: k_2 = 0 \textsf{ or } k_2 \geq 2\}$. Finally, $B$'s information partition stays the same. In the third stage, everybody sends again the message $0$. $B$ splits $\{9,10\}$ into $\{9\}$ and $\{10\}$, whereas $A$ and $C$ do not alter their information partitions. In all subsequent stages indexed by a finite ordinal, this process of communication and learning goes on in a similar fashion. Everybody always sends the message $0$, and, at any stage, either $B$ alone or both $A$ and $C$ can strictly refine their information partitions.

When the stage indexed by the first limit ordinal $\omega$ is reached, each person's information is represented by the coarsest partition that is finer than all the partitions he or she had at previous communication stages. More specifically, we have:
\begin{align*}
g^{\omega}_A &= \left\lbrace \{1,2,7\}, \{3,4\}\right\rbrace \cup \left\lbrace \{5+4k_1\}: k_1 \geq 0\right\rbrace \cup \left\lbrace \{6+4k_2\}:k_2\geq 0\right\rbrace \cup \left\lbrace \{8+4k_3, 11+4k_3\}:k_3\geq 0\right\rbrace\\
g^{\omega}_B &= \left\lbrace \{1, 2\}\right\rbrace \cup \left\lbrace \{3+4k_1, 4+4k_1\}:k_1\geq 0\right\rbrace \cup \left\lbrace \{5+4k_2\}: k_2 \geq 0\right\rbrace \cup \left\lbrace \{6+4k_3\}:k_3\geq 0\right\rbrace\\
g^{\omega}_C &= \left\lbrace \{1+4k_1\}:k_1\geq 0\right\rbrace \cup \left\lbrace\{2+4k_2\}:k_2\geq 0\right\rbrace \cup \left\lbrace \{3+4k_3:k_3\geq 0\}, \{4+4k_4:k_4\geq 0\}\right\rbrace.
\end{align*}

This means that $A$ is able to discriminate between all the states in the first two lines of Figure \ref{fig:Example} except for the states $1$ and $2$, which still belong to $\{1,2,7\}$. Similarly, $B$ can discriminate between all the states in the first two lines except for $1$ and $2$, which still form the block $\{1,2\}$. As to $C$, she can discriminate between all the states in the first two lines of Figure \ref{fig:Example}. Even at this stage, no one knows that the true state of the world is $4$. Consequently, after the first limit ordinal, everyone keeps on sending the message $0$. In the stage indexed by $\omega + 1$, $B$ splits his block $\{1,2\}$ into the two singletons $\{1\}$ and $\{2\}$. In fact, if the true state were $1$ or $2$, $B$ would learn it upon hearing from $C$ at this stage. In the next stage, $A$ splits $\{1,2,7\}$ into $\{1\}$, $\{2\}$ and $\{7\}$, and so on. At the second limit ordinal, all three persons update their information as they did at $\omega$. This leads to the following partitions:
\begin{align*}
g^{\omega\cdot 2}_A &= \left( \{\{k\} : k \geq 1 \} \backslash \{\{3\},\{4\}\} \right) \cup \left\lbrace \{3, 4\}\right \rbrace\\
g^{\omega \cdot 2}_B &= \left( \{\{k\} : k \geq 1 \}\backslash \{\{3\},\{4\}\} \right) \cup \left\lbrace \{3, 4\}\right \rbrace\\
g^{\omega \cdot 2}_C &= \left\lbrace \{k\} : k \geq 1 \right\rbrace.
\end{align*}

At $\omega \cdot 2$, $C$ is finally able to deduce that the true state of the world is $4$. Both $A$ and $B$ are able to discriminate between all states except for $3$ and $4$. Hence they do not know yet if the true state is $3$ or $4$. In the next stage, $C$ sends the message $4$ to $B$, who then sends the same message to $A$. Therefore, at the end of the communication stage indexed by $\omega \cdot 2 +2$, everyone has learned the true state, and the information partitions are:
\begin{equation}\label{eq:Ex-Final-Partitions}
g^{\omega \cdot 2 +2}_A =  g^{\omega \cdot 2 +2}_B= g^{\omega \cdot 2 +2}_C =\left\lbrace \{k\} : k \geq 1 \right\rbrace.
\end{equation}

It follows that, from stage $\omega \cdot 2 +2$ on, everyone will always send the message $4$, and information partitions will clearly remain the same.

Some remarks are in order. First, we can see from \eqref{eq:Ex-Final-Partitions} that everyone eventually shares all his or her information with everyone else. This is not necessarily true in more general cases, and the amount of information shared in a dialogue depends on the expressive power of the function $f$.

Second, everyone agrees before the stage $\omega \cdot 2$, and everyone agrees from $\omega \cdot 2 +2$ on. That is, everyone sends the message $0$ before $\omega \cdot 2$, and everyone sends the message $4$ from $\omega \cdot 2 +2$ on. But there is an important difference between these two cases. When everyone agrees on $0$, there is always a state of the world at which a disagreement would emerge if that state occurred. For example, at the first stage of the dialogue, a disagreement would emerge if the true state were $5$. On the contrary, when everyone agrees on $4$, a consensus is reached not just at the true state but at every possible state of the world. As we show in Section \ref{sec:Results1}, when a consensus holds only at some states, a dialogue enables people to refine their information partitions. On the other hand, when a consensus holds at every possible state, a dialogue does not allow people to refine their information, meaning that there is nothing new to be learned by exchanging messages.

Finally, we remark that the dialogue described so far can also capture observational learning. Specifically, suppose that, in each stage, each person takes an action from the set $D = \mathbb{N} \cup \{0\}$. Stage payoffs are determined by the following function:

\begin{equation*}\label{eq:utilityEx}
	u \left(d, k \right) =
	\begin{cases}
		1 & \text{if } d=k\\
		0 & \text{if } d=0\\
		- 1 & \text{otherwise}.
	\end{cases}
\end{equation*}
In words, taking action $d\neq 0$ in state $k$ yields a reward if the action and the state match. If they don't, the decision maker incurs a loss. In every state, the safe option of choosing $d = 0$ is always available. There are no payoff externalities: the payoff accruing to anyone is independent of what the others do, and vice versa. In each stage, everyone selects an action according to the maximin criterion. For any information block $S \subseteq \mathbb{N}$, the chosen action $\hat{f}(S)$ is
\begin{equation*}
\hat{f}(S) = \max_{d\in D} \min_{k\in S} u(d,k).
\end{equation*}
It is straightforward to verify that $\hat{f}(S) = f(S)$ for every $S$. Now suppose the initial information partitions are again those represented in Figure \ref{fig:Example}. In addition, actions are privately observed. In each stage, $B$ observes the actions chosen by $A$ and $C$, and each of $A$ and $C$ observes $B$'s choice. Assume again that the true state is $4$. In the first stage, everyone chooses $0$. Moreover, $B$ refines his information partition by splitting the block $\{5,6\}$ into $\{5\}$ and $\{6\}$ while $A$ and $C$ do not alter their information partitions. In subsequent stages too, all three persons update their information exactly as they did with $f$. As a consequence, everybody learns the true state at the stage indexed by $\omega \cdot 2 + 2$.

\section{Model}\label{sec:Model1}
\subsection{Setup}
Our object of study is a communication structure $\mathcal{C} = \langle I,X, A, f, G\rangle$, where:
\begin{itemize}
\item  $I=\{1, \dots, n\}$, with $n\geq 2$, is a finite set of agents;
\item $X$ is a non-empty set of states of the world;
\item $A$ is a non-empty set of messages;
\item $f:\mathcal{X} \to A$ is a message function, where $\mathcal{X}$ is the set of non-empty subsets of $X$;
\item $G$ is a directed graph whose set of nodes is $I$. Abusing notation, we write $G$ to indicate both the graph and its set of edges $G\subseteq I\times I$.
\end{itemize}

Agents' information about the true state is represented by partitions. The set of all partitions of $X$ is $\mathcal{P}$, with typical elements $P$, $P'$, etc. Given a state $x\in X$ and a partition $P \in \mathcal{P}$, the block of $P$ containing $x$ is denoted by $P(x)$. The set $\mathcal{P}$ is partially ordered by the binary relation ``is coarser than'', which is denoted by $\leq$ and defined as follows: for any two partitions $P$ and  $P'$, we have $P \leq P'$ if and only if $P$ is a coarsening of $P'$, i.e. every block of $P$ can be written as the union of some blocks of $P'$. We use $P \vee P'$ to denote the join (coarsest common refinement) of $\{P,P'\}$, and $\bigvee \left\lbrace P_h: h \in H\right\rbrace$ for the join of the indexed family $\left\lbrace P_h: h \in H\right\rbrace$. Similarly, we use $\bigwedge \left\lbrace P_h: h \in H\right\rbrace$ to indicate the meet (finest common coarsening) of the family $\left\lbrace P_h: h \in H\right\rbrace$. Recall that $\mathcal{P}$ is a complete lattice.

When agent $i$'s information is represented by $P_i \in \mathcal{P}$, we say that $i$ has information $P_i$. Knowledge is defined in the usual way. Given a state $x\in X$ and an event $E\subseteq X$, we say that agent $i$ knows $E$ in state $x$ if $P_i(x)\subseteq E$. We say that $E$ is common knowledge at $x$ if $\bigwedge \left\lbrace P_i: i \in I\right\rbrace(x)\subseteq E$, where $\bigwedge \left\lbrace P_i: i \in I\right\rbrace(x)$ is the unique block of the meet containing $x$.

\subsection{Messages, communication and learning}
Agents exchange messages in a multi-stage communication process. The message function $f$ determines how agents send messages depending on their information. The graph $G$ determines who sends a message to whom. Both $f$ and $G$ are kept fixed as communication unfolds. We do not make any assumption on the content of messages. In particular, we do not require messages to be posterior beliefs about a fixed event.
\paragraph{Messages.} When $i$ has information $P_i$, we use the function $f_i:X\to A$ to indicate what message $i$ sends at any given state $x$. We refer to $f_i$ as $i$'s message function. Since no confusion should arise, we save on notation and omit the dependence of $f_i$ on $P_i$. We follow \cite{bacharach1985} and make two assumptions about message functions.
\begin{assumption}[Like-mindedness]\label{Assm:L-m} For every $i\in I$, and for every partition $P_i \in \mathcal{P}$, if $i$ has information $P_i$, then $f_i(x) = f(P_i(x))$ for every $x\in X$.
\end{assumption}
Like-mindedness captures the fact that agents share the same view of the world. If any two agents have the same information in a given state, then they must send the same message in that state. Consequently, if different messages are sent, this is solely due to asymmetric information and not to, say, different subjective states or other forms of fundamental disagreement. Notice that, in every state $x$ and for every agent $i$, the message that $i$ sends when $x$ occurs is a function of the smallest event that $i$ knows at $x$, which is $P_i(x)$.

Another implication of Assumption \ref{Assm:L-m} is that, for every $x,x'\in X$, if  $P_i(x)=P_i(x')$, then $f_i(x)=f_i(x')$. This reflects full rationality. If an agent sent different messages in different states belonging to the same information block, then she would realize that these states are not indistinguishable after all and so she would assign them to different information blocks. In addition, every agent always knows the message she is transmitting.

\begin{assumption}[Sure thing principle (STP)]\label{Assm:STP} For any $S\in \mathcal{X}$, and for any partition $\{S_h: h\in H\}$ of $S$, if $f(S_h)=a$ for all $h\in H$ then $f(S)=a$.
\end{assumption}
The STP is also known as \textit{union consistency}\footnote{The message function $f:\mathcal{X}\to A$ is union consistent if $S,S' \in \mathcal{X}$, $S\cap S' = \emptyset$, and $f(S)=f(S')=a$ imply $f(S\cup S')=a$. Strictly speaking, the union consistency of \cite{cave1983} and the STP of \cite{bacharach1985} are equivalent only when information partitions are finite. With possibly infinite partitions, the STP implies union consistency, but the converse is not true.} (\cite{cave1983}). The STP says the following. Take any collection of mutually disjoint events $\{S_h: h\in H\}$. If there exists a message $a \in A$ such that, for every event $S_h$ in the collection, an agent sends $a$ when she knows that the true state is contained in $S_h$, then she must send the same message $a$ when she knows that the true state is in $\cup_{h\in H} S_h$. \cite{bacharach1985} shows that the STP is satisfied, among others, by expected utility maximization and by the maximin criterion \footnote{See \cite{moses1990agreeing} for a critique of the STP in epistemic models, and \cite{samet2010agreeing} and \cite{tarbush2016counterfactuals} for possible ways to address their critique.}. It is easy to verify that the message function $f$ in the example (Section \ref{sec:Example}) satisfies the STP too. \cite{bacharach1985} notes that Hurwicz' $\alpha$-rule does not satisfy the STP when $0 < \alpha < 1$. We remark that the STP might also be violated by ambiguity averse agents. Specifically, suppose an agent has full information and, at every block of her information partition, she announces that her optimal choice is $a$. Under ambiguity aversion, the agent could choose a different action $b$ if she only knew that the state is contained in the union of those blocks\footnote{I thank an anonymous reviewer for suggesting that the STP does not hold under ambiguity aversion.}.

We can now introduce \textit{working partitions} (\cite{weyers2992}). Given a message function $f_i$, we use $W_i$ to denote the corresponding working partition. For every $x\in X$, the block of $W_i$ containing $x$ is $W_i(x):= \{x'\in X : f_i(x')=f_i(x)\}$. In words, $W_i(x)$ is the event ``$i$ sends message $a$'' for some $a\in A$. Therefore, we can interpret $W_i(x)$ as the information that $i$ conveys to any receiver $j\neq i$ by sending message $f_i(x)$ in state $x$. The fact that $W_i$ is a partition reflects the lack of ambiguity about the interpretation of messages. Since no confusion should arise, we save on notation and omit the dependence of $W_i$ on the underlying information partition $P_i$. Finally, notice that $W_i$ is necessarily a coarsening of $P_i$.

\paragraph{Communication  and learning.}
Communication unfolds in stages. We do not put any limit on the length of the communication process. In particular, agents can engage in dialogues of transfinite length. A transfinite protocol is a sequence $\left(G^{\alpha}: \alpha \in \mathsf{Ord} \right)$, where $G^{\alpha} \subseteq I\times I$ for every ordinal $\alpha$. If $(i,j)\in G^{\alpha}$ for some ordinal $\alpha$, there is a directed edge from $i$ to $j$ and this means that $i$ sends a message to $j$ in stage $\alpha$. Throughout the paper, we confine ourselves to protocols such that $G^{\alpha} = G$ for every ordinal $\alpha$. In other words, we consider only stationary protocols induced by the fixed graph $G$.

We are going to use a function $g$ to describe how receivers update their information upon receiving messages. Before formally defining such a function, we introduce a couple of auxiliary objects. Let $\mathcal{P}^n$ be the $n$-fold Cartesian product of $\mathcal{P}$. An element of $\mathcal{P}^n$ is an indexed collection  $\boldsymbol{P} = \left(P_1, \dots, P_n\right)$ of partitions of $X$. The set $\mathcal{P}^n$ is endowed with the following product order:
\begin{equation*}
(P_1, \dots, P_n) \leq (P'_1,\dots, P'_n) \iff P_i \leq P'_i \; \text{ for all } i\in I.
\end{equation*}
Notice that $\mathcal{P}^n$ is a complete lattice too. For each $i\in I$, the (possibly empty) set $S(i):=\left\lbrace j\in I: (j,i)\in G\right\rbrace$ is the subset of agents who send a message to $i$. We are now ready to define the function $g: \mathcal{P}^n \to \mathcal{P}^n$ as follows:
\begin{equation}\label{eq:g}
g_i\left((P_1,\dots,P_n)\right)  =
\begin{cases}
\bigvee \left\lbrace P_i \vee W_j\right\rbrace_{j\in S(i)} & \text{if } S(i)\neq \emptyset\\
P_i & \text{otherwise},
\end{cases}
\end{equation}
where we write $g_i\left((P_1,\dots,P_n)\right)$ to denote the $i$th component of $g\left((P_1,\dots,P_n)\right)$.

In words, $g$ maps profiles of information partitions at the beginning of a given stage of communication to profiles of updated partitions. The updating is done in light of the messages exchanged during that communication stage. If $i$ does not receive any message, her information stays the same. But if she receives a message from $j$, she updates her information by taking the join between her own information partition $P_i$ and $j$'s  working partition $W_j$. If $i$ receives messages from more than one sender, she takes the join between her own information partition and the working partitions of all $j\in S(i)$. The learning process just described is introduced by \cite{parikh-krasucki1990} and then amended by \cite{weyers2992}. In particular, \cite{weyers2992} argues that fully rational agents update not only the partition block containing the true state of the world but their whole information partition.

Given a graph $G$ and a profile $\boldsymbol{P^0}$ of initial information partitions, we define the \textit{dialogue} starting from $\boldsymbol{P^0}$ as the sequence $\left(g^{\alpha}: \alpha \in \mathsf{Ord} \right)$ constructed recursively as follows:
\begin{align*}
g^0&:=\boldsymbol{P^0},\\
g^{\alpha +1}&:=g\left(g^{\alpha}\right) \text{ for every ordinal } \alpha,\\
g^{\lambda} &:= \bigvee \left\lbrace g^{\alpha}: \alpha < \lambda \right\rbrace \text{ for every limit ordinal } \lambda.
\end{align*}

In words, a dialogue is a transfinite sequence in $\mathcal{P}^n$ starting from an initial profile $\boldsymbol{P^0}$ and constructed by iterating the function $g$ induced by $G$. For every ordinal $\alpha$, the element $g^{\alpha}$ of the dialogue tells us what information agents have at the end of the $\alpha$th stage of communication. At every successor ordinal, agents update their information as per \eqref{eq:g}. At every limit ordinal $\lambda$, each agent takes the join of the information partitions she had at all the previous stages $\alpha < \lambda$ of the dialogue.

The initial profile $\boldsymbol{P^0}$ can be thought of as exogenous information, whereas partitions $g^{\alpha}$, with $\alpha > 0$, can be interpreted as endogenous information. Differently put, $\boldsymbol{P^0}$ captures the information conveyed by a privately observed signal about the state. We do not explicitly model this signal and only assume that it induces partitional information. Nature acts only once and chooses the true state of the world $x$, which in turn determines the signal realization $P_i^0(x)$ observed by each $i$. Subsequent information partitions are formed as the above process of communication and learning unfolds.

We define a dialogue as a sequence of partitions and not as a sequence of messages. This allows us to exploit the lattice structure of $\mathcal{P}^n$ when we study dialogues leading to consensus in Section \ref{sec:Results1}. Notice that, for a given message function $f$, a profile $\boldsymbol{P}$ of information partitions fully determines the unique profile of messages transmitted at every state. We conclude this section by emphasizing that the whole structure of the model is common knowledge. In particular, it is commonly known who talks with whom and when, how partition blocks are mapped to messages, and how information is updated.

\section{Results}\label{sec:Results1}
In this section we study conditions under which a rational dialogue eventually leads to consensus. We first give a full characterization of consensus for the static case, i.e. for a fixed profile of information partitions; we then move on to the dynamic case, in which communication is allowed.

\subsection{Consensus}\label{subsec:Consensus}
Let $\boldsymbol{P}\in \mathcal{P}^n$ be a profile of information partitions at a given stage. We say that a \textit{consensus} holds at that stage if, for all $i,j\in I$, we have that $f_i = f_j$. That is, for all $i,j\in I$, $f_i(x)=f_j(x)$ for every state $x \in X$. If agents agree at some state $x$ but not necessarily at every state, then we say that a \textit{partial} consensus holds at $x$. In the example in Section \ref{sec:Example}, the consensus that holds at the true state $4$ is only partial before the stage $\omega \cdot 2$, but then it becomes a full consensus from $\omega \cdot 2 + 2$ on.

Our first result is a characterization of consensus for a given profile of information partitions.
\begin{proposition}\label{prop:Consensus}
Suppose agents have information $\boldsymbol{P}=\left(P_1, \dots, P_n\right)$. Then the following are equivalent:
\begin{enumerate}
\item[a)] For all $x\in X$, the profile of messages that is sent at $x$, i.e. the event \begin{equation*}
E(x)=\left\lbrace x' \in X: f_1(x')= f_1(x), \dots, f_n(x')=f_n(x) \right\rbrace,
\end{equation*}
is common knowledge at $x$;
\item[b)] For all $i,j \in I$, $f_i=f_j$;
\item[c)] For all $i,j \in I$, $W_i=W_j$.
\end{enumerate}
\end{proposition}

\begin{proof}
Firstly, $a) \implies b)$ follows from Theorem 3 in \cite{bacharach1985}. Secondly, $b) \implies c)$ follows immediately from the definition of working partitions. Finally, in order to show $c) \implies a)$, fix a state $x\in X$. By the definition of working partitions, for every $i\in I$, the event $\left\lbrace x' \in X: f_i(x')=f_i(x)\right\rbrace$ is the same as $W_i(x)$. Let $W(x):= \cap_{i\in I} W_i(x)$. Since every $W_i$ is a coarsening of $P_i$, and since $W_i(x)=W_j(x)$ for all $i,j\in I$ by assumption, we have that, for all $i\in I$,
\begin{equation*}
P_i(x) \subseteq W_i(x) = W(x).
\end{equation*}
Therefore, for all $i\in I$,
\begin{equation*}
P_i(x) \subseteq \bigwedge \left\lbrace P_i: i \in I\right\rbrace(x) \subseteq W(x).
\end{equation*}
\end{proof}

Proposition \ref{prop:Consensus} says that consensus is equivalent to an epistemic configuration where, at every state, the profile of messages that are being sent at that state is common knowledge. Furthermore, when messages are commonly known, they must be the same. The latter statement is nothing other than the generalized version of Aumann's agreement theorem established in \cite{bacharach1985}. We can also say that a consensus cannot hold without it being common knowledge that it holds. This is not necessarily true for partial consensus. In the example (Section \ref{sec:Example}), the profile of messages sent at the true state is not common knowledge until a full consensus is reached at the stage $\omega \cdot 2 + 2$.

We remark that consensus does not imply that all agents have the same information partitions. Furthermore, the equivalence in Proposition \ref{prop:Consensus} crucially relies on the STP. As we show in Section \ref{sec:Disc}, if the STP does not hold, one can only conclude that $b)\implies c) \implies a)$. Consequently, it would no longer be impossible to agree to disagree.

The following corollary will prove useful in establishing subsequent results.
\begin{corollary}\label{cor:disagree}
If $f_i\neq f_j$, then $P_i < P_i \vee W_j$ or $P_j < P_j \vee W_i$.
\end{corollary}
\begin{proof}
By contraposition, suppose that neither $P_i < P_i \vee W_j$ nor $P_j < P_j \vee W_i$ hold. Since it is always the case that $P_i \leq P_i \vee W_j$ and $P_j \leq P_j \vee W_i$, we must have both $P_i = P_i \vee W_j$ and $P_j = P_j \vee W_i$. This implies that $W_i$ is a coarsening of $P_j$ and $W_j$ is a coarsening of $P_i$. Combining this with the fact that each working partition is a coarsening of the underlying information partition, we have that both $W_i$ and $W_j$ are common coarsenings of $P_i$ and $P_j$. Therefore, for every $x\in X$, $W_i(x) \cap W_j(x)$ is common knowledge at $x$ between $i$ and $j$. Thus it follows from Proposition \ref{prop:Consensus} that $f_i=f_j$.
\end{proof}
The interpretation is that, if $i$ and $j$ disagree at some state, then it must be the case that either $i$ can strictly refine her information by receiving a message from $j$, or $j$ can strictly refine her information by receiving a message from $i$, or both. In other words, when two agents disagree, at least one of them can learn some new information from the other.

\subsection{Dialogues leading to consensus}\label{subsec:Dialogues}
We now examine conditions under which dialogues lead to consensus. Given a communication graph $G$ and a profile $\boldsymbol{P^0}$ of initial information partitions, we say that the dialogue $\left(g^{\alpha}: \alpha \in \mathsf{Ord}\right)$ induced by $G$ and starting from $\boldsymbol{P^0}$ leads to a consensus if, for every $i,j\in I$, we have that $f_i^{\alpha}=f_j^{\alpha}$ for some ordinal $\alpha$, where $f_i^{\alpha}$ is $i$'s message function associated with $g_i^{\alpha}$. Differently put, a dialogue leads to consensus if the sequence $\left(g^{\alpha}: \alpha \in \mathsf{Ord}\right)$ contains a profile $g^{\alpha} \in \mathcal{P}^n$ at which everybody agrees.

It is clear that the properties of $G$ affect the function $g$ and, consequently, the sequence $\left(g^{\alpha}:\alpha \in \mathsf{Ord}\right)$. We make the following preliminary observation.
\begin{remark*}
For any $G$, the function $g$ is inflationary but need not be monotone\footnote{Let $P$ be a poset. A function $h:P \to P$ is monotone (or order-preserving) if $x \leq y \implies h(x) \leq h(y)$. We say that $h$ is inflationary (or increasing) if, for all $x\in P$, $x\leq h(x)$.}.
\end{remark*}
\begin{proof}
It follows immediately from \eqref{eq:g} that $g$ is inflationary. The following example shows that $g$ need not be monotone. Let $X=\left\lbrace x,y,w,z\right\rbrace$, $I=\{1,2\}$, and let the message function $f$ be such that $f(\{x\})=f(\{x,y\})=a$, and $f(S) = b$ for any other non-empty subset $S$ of $X$. In addition, suppose the communication graph is $G=\{(1,2), (2,1)\}$. Now take the following elements of $\mathcal{P}^2$:
\begin{align*}
\boldsymbol{P} &= \left(P_1,P_2\right) = \left(\{X\}, \{\{x,y\}, \{w,z\}\}\right)\\
\boldsymbol{P'} &= \left(P'_1,P'_2\right) = \left(\{X\}, \{\{x\}, \{y\}, \{w\}, \{z\}\}\right).
\end{align*}
Thus we have
\begin{align*}
g\left(\boldsymbol{P}\right) &= \left(\{\{x,y\}, \{w,z\}\}, \{\{x,y\}, \{w,z\}\}\right)\\
g\left(\boldsymbol{P'}\right) &= \left(\{\{x\}, \{y,w,z\}\}, \{\{x\}, \{y\}, \{w\}, \{z\}\}\right).
\end{align*}
Therefore, $\boldsymbol{P} \leq \boldsymbol{P'}$ but $g\left(\boldsymbol{P}\right) \not\leq g\left(\boldsymbol{P'}\right)$.
\end{proof}

Since $g$ is inflationary, the sequence $(g^{\alpha}: \alpha \in \mathsf{Ord})$ is increasing. That is, for every $\alpha, \beta \in \mathsf{Ord}$, we have that $\beta < \alpha$ implies $g^{\beta} \leq g^{\alpha}$. This can easily be proved by induction on $\alpha$.

We now introduce the following two subsets of $\mathcal{P}^n$:
\begin{equation*}
\mathsf{Cons}(f):= \left\lbrace \boldsymbol{P}\in \mathcal{P}^n: f_i=f_j \text{ for all } i,j\in I\right\rbrace
\end{equation*}
and
\begin{equation*}
\mathsf{Fix}(g):= \left\lbrace \boldsymbol{P}\in \mathcal{P}^n: g(\boldsymbol{P})=\boldsymbol{P} \right\rbrace.
\end{equation*}

In words, $\mathsf{Cons}(f)$ is the set of partition profiles at which a consensus holds, whereas $\mathsf{Fix}(g)$ is the set of fixed points of $g$. 
\begin{proposition}\label{prop:anyG}
For any $G$, we have $\emptyset \neq \mathsf{Cons}(f) \subseteq \mathsf{Fix}(g)$.
\end{proposition}
\begin{proof}
To show that $\mathsf{Cons}(f)$ is non-empty, take any profile $\left(P_1,\dots, P_n\right) \in \mathcal{P}^n$ such that $P_i=P_j$ for every $i,j\in I$. By like-mindedness, all such profiles are contained in $\mathsf{Cons}(f)$.

To show the inclusion $\mathsf{Cons}(f) \subseteq \mathsf{Fix}(g)$, take $\left(P_1,\dots, P_n\right) \in \mathsf{Cons}(f)$. By Proposition \ref{prop:Consensus}, $\left(P_1,\dots, P_n\right)$ is such that $W_i = W_j$ for any $i,j \in I$. Therefore, since $W_i$ is a coarsening of $P_i$, we have that $P_i \vee W_j = P_i$ for any $i,j \in I$. Thus $\left(P_1,\dots, P_n\right)\in \mathsf{Fix}(g)$.
\end{proof}

Since $g$ is inflationary, the non-emptiness of $\mathsf{Fix}(g)$ can also be proved by invoking the Bourbaki-Witt fixed point theorem.

Now, it is clear that if a dialogue $\left(g^{\alpha} : \alpha \in \mathsf{Ord}\right)$ contains a fixed point at $g^\alpha$, then it stays constant at any $\beta > \alpha$. However, it is not necessarily the case that such a dialogue leads to consensus. In order for this to be the case, we need to make sure that the communication process induced by $G$ is sufficiently rich. We thus make the following assumption.
\begin{assumption}\label{Assm:graph}
The communication graph $G$ contains a spanning subgraph\footnote{Recall that a spanning subgraph of $G$ is a subgraph $G'\subseteq G$ with the same set of vertexes as $G$.} $G'$ such that:
\begin{enumerate}
\item[a)] $G'$ is strongly connected: for every distinct $i,j \in I$, there exists a directed path in $G'$ from $i$ to $j$ and a directed path from $j$ to $i$;
\item[b)] $G'$ is symmetric: for every $i,j\in I$, if $(i,j)\in G'$, then $(j,i)\in G'$.
\end{enumerate}
\end{assumption}
Strong connectedness says that no one is excluded from communication, i.e. everyone communicates with everybody else, either directly or indirectly. Symmetry means that communication is reciprocal. When these two conditions are met, the following equivalence holds.
\begin{proposition}\label{prop:equiv}
If $G$ satisfies Assumption \ref{Assm:graph}, then $\mathsf{Cons}(f) = \mathsf{Fix}(g)$.
\end{proposition}
\begin{proof}
By Proposition \ref{prop:anyG}, it is enough to show that $\mathsf{Fix}(g) \subseteq \mathsf{Cons}(f)$. Let $\boldsymbol{P}=\left(P_1,\dots, P_n\right)$ be a fixed point of $g$. Suppose by way of contradiction that $\boldsymbol{P}$ does not induce a consensus. Hence there are distinct $i$ and $j$ in $I$ such that $f_{i}\neq f_{j}$. By strong connectedness, there exists a directed path in $G'\subseteq G$ from $i$ to $j$: that is, for some integer $K\geq 1$, there is a path $i_0, i_1, \dots, i_K$ in $G'$ such that $i_0=i$ and $i_K=j$. Since $i$ and $j$ disagree, this path must contain an edge $(i_k,i_{k+1})$ such that $i_k$ and $i_{k+1}$ disagree, for some $k\in\{0,\dots, K-1\}$. By symmetry, $(i_{k+1},i_k)\in G'$. By Corollary \ref{cor:disagree}, we have 
\begin{equation*}
P_k < P_k \vee W_{k+1} \;\text{ or }\; P_{k+1} < P_{k+1} \vee W_{k},
\end{equation*}
and using this in \eqref{eq:g} we obtain
\begin{equation*}
P_k < g_k(\boldsymbol{P}) \;\text{ or }\; P_{k+1} < g_{k+1}(\boldsymbol{P}),
\end{equation*}
so contradicting the hypothesis that $\boldsymbol{P}$ is a fixed point of $g$.
\end{proof}

In light of Proposition \ref{prop:equiv}, looking for dialogues leading to consensus is the same as looking for fixed points of $g$. Intuitively, we know from Corollary \ref{cor:disagree} that, in case of disagreement between $i$ and $j$, learning can take place in either direction. Assumption \ref{Assm:graph} makes sure that communication between $i$ and $j$ is reciprocal, so that it can never be the case that $i$ and $j$ disagree without having the possibility of exchanging messages between each other. The importance of reciprocity in communication is pointed out by \cite{krasucki1996}, and Example 2 in \cite{parikh-krasucki1990} shows how a consensus may never emerge if one dispenses with it.

We are now ready to state our main result. In what follows, 
we write $\alpha^*$ to denote the least ordinal $\alpha$ such that $g^{\alpha + 1} = g^{\alpha}$.
\begin{theorem*}
Let $\mathcal{C}=\langle I,X, A, f, G\rangle$ be a communication structure satisfying Assumptions \ref{Assm:L-m}-\ref{Assm:graph}. For any profile $\boldsymbol{P^0}$ of initial information partitions, the dialogue $\left(g^{\alpha}:\alpha \in \mathsf{Ord}\right)$ induced by $G$ and starting from $\boldsymbol{P^0}$ leads to a consensus. Furthermore, $\vert \alpha^* \vert\leq n\vert X\vert$.
\end{theorem*}
\begin{proof}
Since $g$ is inflationary, and since $\mathcal{P}^n$ is a complete lattice, by following the same argument as in the proof of Theorem 12.9 in \cite{roman2008lattices} we have that the sequence $\left(g^{\alpha}:\alpha \in \mathsf{Ord}\right)$ starting from $\boldsymbol{P^0}$ is always well-defined, increasing, and contains one, and only one, fixed point  of $g$. By Proposition \ref{prop:equiv}, the dialogue $\left(g^{\alpha}:\alpha \in \mathsf{Ord}\right)$ induces a consensus.

In order to show that $\vert \alpha^* \vert\leq n\vert X\vert$, take the subsequence $\left(g^{\alpha}:\alpha \leq \alpha^* \right)$. Since $\left(g^{\alpha}:\alpha \in \mathsf{Ord}\right)$ is increasing and $g^{\alpha^*}$ is a fixed point of $g$, the subsequence $\left(g^{\alpha}:\alpha \leq \alpha^* \right)$ is strictly increasing, i.e. for all $\alpha, \beta \leq \alpha^*$,
\begin{equation}\label{eq:strict-ref}
\beta < \alpha \implies g^{\beta} \leq g^{\alpha} \; \text{ and } \; g^{\beta} \neq g^{\alpha}.
\end{equation}
Now define the image of $\left(g^{\alpha}:\alpha \leq \alpha^* \right)$ as $\mathsf{Im}:=\left\lbrace g^{\alpha}: \alpha \leq \alpha^*\right\rbrace$. Since $\left(g^{\alpha}:\alpha \leq \alpha^* \right)$ is strictly increasing, $\mathsf{Im}$ is a well-ordered subset of $\mathcal{P}^n$ having order type $\alpha^* +1$. Furthermore, for every $i\in I$, let $\mathsf{Im}_i:=\left\lbrace g_i^{\alpha}: \alpha \leq \alpha^*\right\rbrace$ be the $i$th projection of $\mathsf{Im}$. Notice that $\mathsf{Im}_i$ is a well-ordered chain in $\mathcal{P}$ having order type at most $\alpha^*+1$. By Lemma 3.1 in \cite{avery2018chains}, we have that $\vert \mathsf{Im}_i\vert \leq \vert X\vert$. Now construct a map $\phi : \alpha^* +1 \to \bigsqcup_{i=1}^n \mathsf{Im}_i$, where $\bigsqcup$ denotes disjoint union, as follows. Let $\phi(0) := P^0_1$, and for every $0<\alpha \leq \alpha^*$,
\begin{equation}\label{eq:phi}
\phi(\alpha):= \min_i\left\lbrace g_i^{\alpha}: g_i^{\beta}\leq g_i^{\alpha} \;\text{ and }\; g_i^{\beta}\neq g_i^{\alpha} \; \text{ for all } \; \beta < \alpha\right\rbrace.
\end{equation}
In words, $\phi$ maps each ordinal $\alpha$ less than or equal to $\alpha^*$ to an individual partition $g_i^{\alpha}$ that is a strict refinement of all partitions $g_i^{\beta}$ having index less than $\alpha$. Without loss of generality, in case of multiple individual partitions satisfying \eqref{eq:phi}, we take the one with the lowest (agent) index. It follows from \eqref{eq:strict-ref} that $\phi$ is well-defined and injective. Therefore we have
\begin{equation*}
	\vert \alpha^* \vert \leq \vert \alpha^*+1\vert \leq \Big\vert \bigsqcup_{i=1}^n \mathsf{Im}_i\Big\vert \leq n\vert X\vert.
\end{equation*}
\end{proof}

Clearly, when $X$ is an infinite set, we have $\vert \alpha^* \vert\leq n\vert X\vert = \vert X\vert$.

\section{Discussion}\label{sec:Disc}
\subsection{Related literature}
The paper contributes to the vast literature on common knowledge and agreement initiated by \cite{aumann1976agreeing}, surveys of which can be found in \cite{bonanno1997agreeing} and \cite[Chapter 2]{menager2006thesis}. In particular, we contribute to the strand that investigates rational dialogues and convergence to consensus. \cite{krasucki1996} and \cite{parikh1992infinite} are the closest papers, and below we briefly discuss the connection with them.

\cite{geanakoplos1982} introduce dialogues in a two-agent model with a finite state space where the messages exchanged are posterior beliefs about a fixed event. \cite{bacharach1985} and \cite{cave1983} show that a consensus can be reached not only when people communicate posterior beliefs but also when they communicate the values of any function satisfying a condition akin to the sure thing principle from decision theory. \cite{bacharach1985} considers a two-agent model in which initial information partitions are finite, whereas \cite{cave1983} assumes that the state space is a measure space and that communication is public. \cite{nielsen1984}, \cite{washburn1984asymptotic} and  \cite{bergin1989} study convergence to consensus but they all confine their attention to the case where the state space is a probability space.

Dialogues with private communication between more than two agents are introduced by \cite{parikh-krasucki1990} and further examined in \cite{krasucki1996}, \cite{heifetz1996} and \cite{menager2006consensus}. They all consider models with finite information partitions. The model in our paper is the same as that in \cite{krasucki1996} except for two crucial aspects. While he considers finite partitions and dialogues of finite length, we do not make any of these finiteness assumptions.

Dialogues of transfinite length are studied in \cite{aumann-hart2003} and \cite{parikh1992infinite}. Contrary to this paper, \cite{aumann-hart2003} consider problems in which a simultaneous-move game is played after countably many cheap talk messages have been exchanged. Like this paper, \cite{parikh1992infinite} studies common knowledge acquisition and consensus. However, there are substantial differences with our paper. Most importantly, his model uses Kripke structures and has two agents and a state space that is at most countably infinite. In our paper, we use Aumann (or partition) structures, we have many agents, and we allow the state space to have any cardinality. We remark that, while there is scant literature on transfinite dialogues, there are many papers in which the analysis of interactive knowledge is done on the transfinite domain but without any form of communication between agents. For example, see \cite{lipman1994} and \cite{heifetz-samet1998knowledge}.  

Finally, this paper is also related to \cite{mueller-frank2013}, in which a framework for learning in social networks is provided. The main difference with our paper is that only dialogues having order type $\omega$ are considered by \cite{mueller-frank2013}. As a consequence, his notion of convergence when information partitions are infinite is different from ours.

\subsection{The Sure Thing Principle} Proposition \ref{prop:Consensus} hinges upon the STP. If one dispenses with it, the equivalence $a) \iff b) \iff c)$ breaks down and one can only conclude that $b) \implies c) \implies a)$. Notice that we do not use the STP when we show that $b) \implies c)$ and $c) \implies a)$ in the proof of Proposition \ref{prop:Consensus}. In the following two examples, we show that the converse implications do not necessarily hold. First, let us show that $b)$ does not follow from $a)$ or from $c)$. Set $X=\left\lbrace x,y \right\rbrace$ and let the message function be such that $f(\{x\})=f(\{y\})=a$ and $f(\{x,y\})=b$. Suppose that there are two agents whose information partitions are $P_1 = \left\lbrace \{x\}, \{y\}\right\rbrace$ and $P_2 = \left\lbrace \{x,y\}\right\rbrace$. At every state, it is common knowledge what messages $1$ and $2$ are sending, but clearly $f_1 \neq f_2$. Notice that we also have $W_1 = W_2 = \{X\}$.

Now we show that $c)$ does not follow from $a)$. Set $X=\{x,y,z\}$ and let the message function be such that
\begin{align*}
f(\{x\})=f(\{y\})&=a\\
f(\{z\})=f(\{x,y\})&=b.
\end{align*}
Suppose there are two agents with the following information partitions:
\begin{align*}
P_1&=\left\lbrace \{x\},\{y\}, \{z\}\right\rbrace\\
P_2&=\left\lbrace \{x,y\}, \{z\}\right\rbrace.
\end{align*}
At every state, the profile of messages is common knowledge but $W_1\neq W_2$.

\subsection{Common knowledge of the communication protocol} Our analysis assumes that the way in which a dialogue unfolds is common knowledge. In particular, the communication protocol induced by the graph $G$ is commonly known, and this is crucial in order to have a well-defined learning process. When $i$ receives a message from $j$, she knows exactly whom $j$ talked with in the past and, consequently, she can infer what information $j$ learned from that history even if she does not necessarily know the actual message that $j$ sent to others or received in some state. Differently put, a commonly known communication protocol implies that the informational content of any given message is not ambiguous, so making it possible for people to learn. In the case in which such a common knowledge assumption is relaxed, learning is not well-defined in our framework and convergence to consensus is not guaranteed. In these cases, \cite{koessler2001} and \cite{tsakas-voorneveld2011} show that one needs to enlarge the state space so as to include any possible history of communication. In so doing, uncertainty about the communication structure can be dealt with in the enlarged state space. The fact that we keep the state space fixed throughout a dialogue is a direct consequence of having a commonly known communication protocol.

We also emphasize that communication takes place through a faultless and fully reliable channel. To wit, when $i$ sends a message to $j$, the message is delivered to $j$ with absolute certainty, and it is common knowledge that it is so. We rule out the possibility that a message never reaches the intended recipient and also the possibility that a recipient gets a different message than what was sent by the sender. Should the communication channel be unreliable, we would be in a situation akin to the email game of \cite{rubinstein1989email}, where convergence to consensus is not guaranteed to hold.

\subsection{Non-strategic communication} We assume that communication is not strategic. When strategic motives are introduced, a consensus is not always reached. In a different yet related setting, \cite{anderlini-et-al2011communication} show that, while agents with common interests are able to aggregate their information in a full learning equilibrium, no such an equilibrium can be sustained when interests diverge. However, \cite{ostrovsky2012information}\footnote{I thank an anonymous reviewer for bringing my attention to this paper.} shows that, for a specific class of securities, strategic incentives do not preclude traders in dynamic financial markets from aggregating their private information and reaching a consensus.

\subsection{Necessity of Assumption \ref{Assm:graph}} Strong connectedness and symmetry in Assumption \ref{Assm:graph} are sufficient for the emergence of consensus in a dialogue but they are not necessary. For example, consider an initial profile $\boldsymbol{P^0}$ such that the individual partitions are ordered as follows: $P^0_1 \leq P^0_{2} \leq \cdots \leq P^0_n$. A consensus can be obtained after only one round of communication by letting $n$, who is the most informed agent, send a message to everybody else.

\subsection{Information sharing}
When a consensus is reached, agents do not always have the same information. But if the message function has enough expressive power, then everyone ends up having the same information when a consensus holds. More specifically, suppose the message function $f$ is injective\footnote{I thank an anonymous reviewer for suggesting the connection between injective message functions and information sharing.}. Clearly, the STP is vacuously satisfied. Moreover, each agent's information partition $P_i$ is the same as her working partition $W_i$. And since everyone has the same working partition when a consensus is reached, everyone must have the same information partition too.

\bibliographystyle{plainnat}
\bibliography{biblio}

\end{document}